\documentclass[preprint]{elsarticle}
\usepackage{amssymb, amscd, amsfonts, amsthm, amsmath}
\usepackage[mathscr]{eucal}
\usepackage{subfigure}
\usepackage{lscape}
\usepackage{caption}

\newtheorem{theorem}{Theorem}[section]
\newtheorem{lemma}[theorem]{Lemma}

\newtheorem{proposition}[theorem]{Proposition}

\theoremstyle{definition}
\newtheorem{definition}[theorem]{Definition}

\newtheorem{claim}{Claim}

\theoremstyle{remark}
\newtheorem*{remark}{Remark}

\newcommand{\remove}[1]{}

\newcommand{\pa}[2]{\mathbf{P}_{#1}^{(#2)}}
\newcommand{\cy}[2]{\mathbf{C}_{#1}^{(#2)}}
\newcommand{\hp}[2]{\mathbf{H}_{#1}^{(#2)}}
\newcommand{\hc}[2]{\mathbf{M}_{#1}^{(#2)}}
\newcommand{\G}{\mathbf{G}}

\newcommand{\eg}{\textit{e.g.}}
\newcommand{\ie}{\textit{i.e.}}

\begin{document}

\title{Generalized Fibonacci and Lucas cubes arising from powers of paths and cycles}


\author[di]{P.~Codara\corref{cor1}}
\ead{codara@di.unimi.it}
\author[di]{O.M.~D'Antona}
\ead{dantona@di.unimi.it}
\cortext[cor1]{Corresponding author}
\address[di]{Dipartimento di Informatica,
via Comelico 39, I-20135 Milano, Italy}




\begin{abstract}
The paper deals with some generalizations of Fibonacci and Lucas sequences,
arising from powers of paths and cycles, respectively.

In the first part of the work we provide a formula for the number of edges of
the Hasse diagram of the independent sets of
the $h$\textsuperscript{th} power of a path ordered by
inclusion. For $h=1$ such a diagram is called a Fibonacci
cube, and for $h>1$ we obtain a generalization of the Fibonacci cube.
Consequently, we derive a generalized notion of Fibonacci sequence,
called $h$-Fibonacci sequence.
Then, we show that the number of edges of a generalized Fibonacci cube
is obtained by convolution of an $h$-Fibonacci sequence with itself.

In the second part we consider the case of cycles.
We evaluate the number of edges of the Hasse diagram of the independent
sets of the $h$\textsuperscript{th} power of a cycle ordered by
inclusion. For $h=1$ such a diagram is called Lucas
cube, and for $h>1$ we obtain a generalization of the Lucas cube.
We derive then a generalized version of the Lucas sequence,
called $h$-Lucas sequence. Finally, we show that the number of edges
of a generalized Lucas cube is obtained by an appropriate convolution
of an $h$-Fibonacci sequence with an $h$-Lucas sequence.
\end{abstract}

\maketitle

\section{Introduction}
For a graph $\G$ we denote by $V(\G)$ the set of its vertices, and by $E(\G)$
the set of its edges.
\begin{definition}\label{def:h-path}\label{def:h-cycle}
For $n, h\geq0$,
\begin{itemize}
\item[(i)] the \emph{$h$-power of a path}, denoted by $\pa{n}{h}$, is a graph
with $n$ vertices $v_{1}$, $v_{2}$, $\dots$, $v_{n}$ such that, for $1\leq i,j\leq n$, $i\neq j$,
$(v_i,v_j)\in E(\pa{n}{h})$ if and only if $|j-i|\leq h$;
\item[(ii)] the \emph{$h$-power of a cycle}, denoted by $\cy{n}{h}$, is a graph
with $n$ vertices $v_{1}$, $v_{2}$, $\dots$, $v_{n}$ such that, for $1\leq i,j \leq n$, $i\neq j$,
$(v_i,v_j)\in E(\cy{n}{h})$ if and only if $|j-i|\leq h$ or $|j-i|\geq n-h$.
\end{itemize}
\end{definition}
Thus, for instance, $\pa{n}{0}$ and $\cy{n}{0}$ are the graphs made of $n$ isolated nodes,
$\pa{n}{1}$ is the path with $n$ vertices, and  $\cy{n}{1}$ is the cycle with $n$ vertices.
Figure  \ref{fig:PQ_1-6_2} shows
some powers of paths and cycles.

\begin{figure}[h]
 \centering
 \subfigure[The graphs $\pa{1}{2}, \dots, \pa{5}{2}$]
   {\label{fig:P_1-6_2}\includegraphics[scale=0.7]{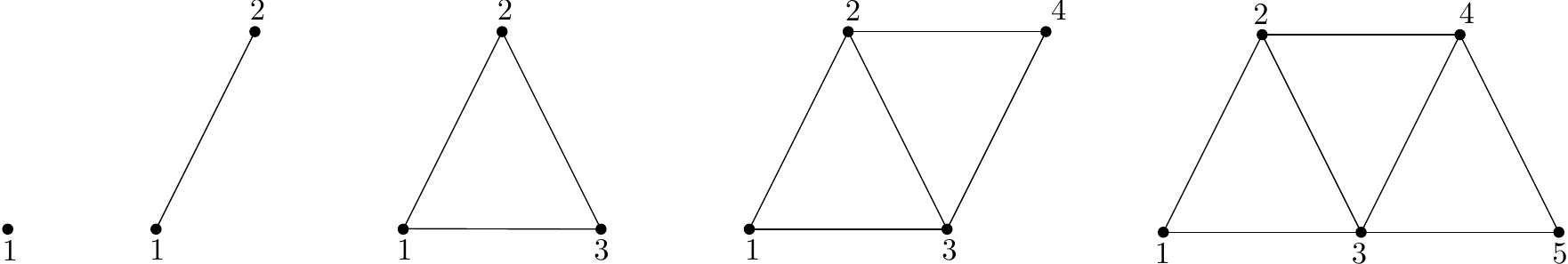}}\\
 \subfigure[The graphs $\cy{1}{2}, \dots, \cy{5}{2}$]
   {\label{fig:Q_1-6_2}\includegraphics[scale=0.75]{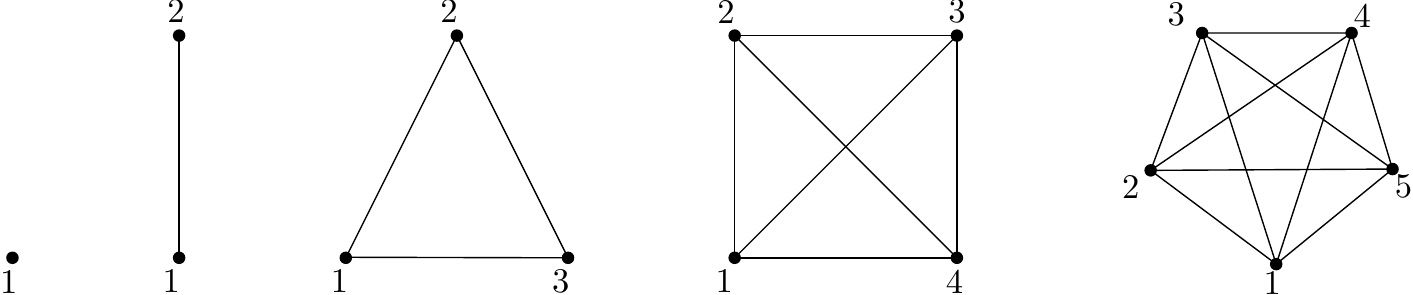}}
 \caption{Some powers of paths and cycles.}
 \label{fig:PQ_1-6_2}
 \end{figure}

\begin{definition}\label{def:indipendent subset}
An \emph{independent set of a graph $\G$} is a subset of $V(\G)$
not containing adjacent vertices.
\end{definition}

Let $\hp{n}{h}$, and $\hc{n}{h}$ be the Hasse diagrams  of the posets of independent sets of $\pa{n}{h}$, and $\cy{n}{h}$, respectively,
ordered by inclusion. Clearly, $\hp{n}{0}\cong\hc{n}{0}$ is a Boolean lattice with $n$ atoms ($n$-cube, for short).

\medskip
Before introducing the main results of the paper, we now
provide some background on Fibonacci and Lucas cubes.
Every independent set
$S$ of $\pa{n}{h}$ can be represented by a binary string $b_1 b_2 \cdots b_n$,
where, for $i\in\{1,\dots,n\}$, $b_i = 1$ if and only if
$v_i \in S$.
Specifically, each
independent set of $\pa{n}{h}$ is associated with a binary string of length
$n$ such that the distance between any two $1$'s of the string is greater than $h$.
Following \cite{fibocubes_enum} (see also \cite{survey_fibo}),
a \emph{Fibonacci string of order $n$} is a binary strings of length $n$ without
consecutive $1$'s. Recalling that the Hamming distance between two binary strings $\alpha$ and $\beta$
is the number $H(\alpha,\beta)$ of bits where $\alpha$ and $\beta$ differ, we can define
the \emph{Fibonacci cube of order $n$}, denoted $\Gamma_n$, as the graph $(V,E)$, where $V$ is the set of all
Fibonacci strings of order $n$ and, for all $\alpha,\beta\in V$, $(\alpha,\beta)\in E$ if and only if $H(\alpha,\beta)=1$.
One can observe that for $h=1$ the binary strings associated with independent sets of $\pa{n}{h}$
are \emph{Fibonacci strings of order $n$}, and the Hasse diagram of the set of all such strings ordered
bitwise (\ie, for $S=b_1 b_2 \cdots b_n$ and
$T=c_1 c_2 \cdots c_n$, $S\geq T$ if and only if
$b_i\geq c_i$, for every $i\in\{1,\dots,n\}$)
is $\Gamma_n$. Fibonacci cubes were
introduced as an interconnection scheme for multicomputers in \cite{hsu}, and their
combinatorial structure has been further investigated, \eg\, in \cite{klavzar,fibocubes_enum}.
Several generalizations of the notion of Fibonacci cubes has been proposed
(see, \eg, \cite{gen_fibo,survey_fibo}).

\begin{remark}Consider the \emph{generalized Fibonacci cubes} described in \cite{gen_fibo}, \ie, the graphs $\mathbf{B}_n(\alpha)$ obtained from the $n$-cube $\mathbf{B}_n$ of all binary strings of length $n$ by removing all vertices that contain the binary string $\alpha$ as a substring. In this notation
the Fibonacci cube is $\mathbf{B}_n(11)$. It is not difficult to see that $\hp{n}{h}$ cannot be expressed, in general, in terms of
$\mathbf{B}_n(\alpha)$. Instead we have:

\centerline{
$\hp{n}{2}=\mathbf{B}_n(11)\cap \mathbf{B}_n(101)\,,\
\hp{n}{3}=\mathbf{B}_n(11)\cap \mathbf{B}_n(101)\cap \mathbf{B}_n(1001)\,,\ \dots\,,$}
\noindent where $\mathbf{B}_n(\alpha)\cap \mathbf{B}_n(\beta)$ is the subgraph of $\mathbf{B}_n$ obtained by removing
all strings that contain either $\alpha$ or $\beta$.
\end{remark}

A similar argument can be carried out for cycles. Indeed,
every independent set
$S$ of $\cy{n}{h}$ can be represented by a circular binary string (\ie, a sequence of $0$'s and $1$'s with the first and last bits considered to be adjacent) $b_1 b_2 \cdots b_n$,
where, for $i\in\{1,\dots,n\}$, $b_i = 1$ if and only if
$v_i \in S$. Thus, each
independent set of $\cy{n}{h}$ is associated with a circular binary string of length
$n$ such that the distance between any two $1$'s of the string is greater than $h$.
A \emph{Lucas cube of order $n$}, denoted $\Lambda_n$, is defined as the graph whose vertices
are the binary strings of length $n$ without either two consecutive $1$'s or a $1$ in the first
and in the last position, and in which the vertices are adjacent when their Hamming distance
is exactly $1$ (see \cite{lucas}). For $h=1$ the Hasse diagram of the set of all circular binary strings
associated with independent sets of $\cy{n}{h}$ ordered bitwise is $\Lambda_n$.
A generalization of the notion of Lucas cubes has been proposed
in \cite{gen_lucas}.
\begin{remark}Consider the \emph{generalized Lucas cubes} described in \cite{gen_lucas}, that is, the graphs $\mathbf{B}_n(\widehat{\alpha})$ obtained from the $n$-cube $\mathbf{B}_n$ of all binary strings of length $n$ by removing all vertices that have a \emph{circular containing} $\alpha$ as a substring (\ie, such that $\alpha$ is contained in the circular binary strings obtained by connecting first and last bits of the string). In this notation
the Lucas cube is $\mathbf{B}_n(\widehat{11})$. It is not difficult to see that $\hc{n}{h}$ cannot be expressed, in general, in terms of
$\mathbf{B}_n(\widehat{\alpha})$. Instead we have:

\smallskip\centerline{
$\hc{n}{2}=\mathbf{B}_n(\widehat{11})\cap \mathbf{B}_n(\widehat{101})\,,\
\hc{n}{3}=\mathbf{B}_n(\widehat{11})\cap \mathbf{B}_n(\widehat{101})\cap \mathbf{B}_n(\widehat{1001})\,,\ \dots$}
\end{remark}

To the best of our knowledge, our $\hp{n}{h}$, and $\hc{n}{h}$ are new generalizations of Fibonacci and Lucas cubes, respectively.

\medskip
In the first part of this paper (which is an extended version of
\cite{endm} --- see remark at the end of this section)
we evaluate $p_n^{(h)}$, \ie, the number of \emph{independent sets}
of $\pa{n}{h}$, and $H_n^{(h)}$, \ie, the number of edges of $\hp{n}{h}$.
We then introduce a generalization of the Fibonacci sequence, that we call
\emph{$h$-Fibonacci sequence} and denote by $\mathcal{F}^{(h)}$.
Such integer sequence is based on the values of $p_n^{(h)}$.
Our main result (Theorem \ref{th:main}) is that,
for $n,h\geq 0$, the sequence $H_n^{(h)}$ is obtained
by convolving the sequence $\mathcal{F}^{(h)}$
with itself.

\medskip
In the second part we deal with power of cycles, and derive similar  results for this case.
Specifically, we compute $c_n^{(h)}$, \ie, the number of \emph{independent sets}
of $\cy{n}{h}$, and $M_n^{(h)}$, \ie, the number of edges of $\hc{n}{h}$.
Further, we introduce a generalization of the Lucas sequence, that we call
\emph{$h$-Lucas sequence} and denote by $\mathcal{L}^{(h)}$.
Such integer sequence is based on the values of $c_n^{(h)}$.
The analogous of Theorem \ref{th:main} in the Lucas case
(Theorem \ref{th:main_cycles}) states that,
for $n>h\geq 0$, the sequence $M_n^{(h)}$ is obtained
by an appropriate convolution between the sequences $\mathcal{F}^{(h)}$
and $\mathcal{L}^{(h)}$.

\begin{remark}
As mentioned before, this work is an extended version of \cite{endm}.
(The extended abstract \cite{endm} has been presented at the conference
Combinatorics 2012, Perugia, Italy, 2012, and consequently published in ENDM.) Specifically,
Section \ref{sec:ind sub of paths} of the present paper is \cite[Section 2]{endm} enriched with
some details and full proofs, while Section \ref{sec:hasse of paths} is \cite[Section 3]{endm}, with
corrected notation and some remarks added.
Lemma \ref{lem:T formula}, Lemma \ref{lem:somma tabelle T}, and Theorem \ref{th:main}
have remained as they were in \cite{endm}: we decided to keep full proofs of these results, to make
the content clearer and in order for the paper to be self-contained.
\end{remark}

\section{The independent sets of powers of paths}\label{sec:ind sub of paths}

For $n, h, k\geq 0$, we denote by $p_{n,k}^{(h)}$ the number of  independent
$k$-subsets of $\pa{n}{h}$.
\begin{remark}
$p_{n,k}^{(1)}$ counts the number of binary strings $\alpha\in \Gamma_n$
with $k$ $1$'s.
\end{remark}
\begin{lemma}
\label{lem:p_nk formula}
For $n, h, k \geq 0$,
\[
p_{n,k}^{(h)} = \binom{n-hk+h}{k}\ .
\]
\end{lemma}
This is Theorem 1 of \cite{hoggatt}. An alternative proof follows.
\begin{proof}
By Definition \ref{def:indipendent subset},
any two elements $v_i,v_j$ of an independent set of $\pa{n}{h}$
must satisfy $|j-i|>h$. It is straightforward to check that
whenever $n- hk+h<0$, $p_{n,k}^{(h)} = 0 = \binom{n-hk+h}{k}$.
It is also immediate to see that when $n=h=0$ our lemma holds true.

Suppose now $n- hk+h\geq 0$. We can complete the proof of our lemma by establishing a bijection
between independent $k$-subset of $\pa{n}{h}$ and $k$-subsets
of a set with $(n-hk+h)$ elements. Let $\mathscr{K}$ be the set of all $k$-subsets of a set $B = \{b_{1},b_{2},\dots,b_{n-hk+h}\}$, and
$\mathcal{I}_k$ the set of all independent $k$-subsets of $\pa{n}{h}$.
Consider the map $f: \mathscr{K} \to \mathcal{I}_k$ such that, for any $S=\{b_{i_1},b_{i_2},\dots, b_{i_k}\} \in \mathscr{K}$, with $1\leq i_1<i_2<\cdots <i_k \leq n-hk+h$,
\[
f(\{b_{i_1},b_{i_2},\dots, b_{i_j}, \dots, b_{i_k}\}) = \{v_{i_1},v_{i_2 + h},\dots, v_{i_j + (j-1)h}, \dots, v_{i_k + (k-1)h}\}\,.
\]

\begin{claim}\label{claim:1} The map $f$ associates an independent $k$-subset of $\pa{n}{h}$
with each $k$-subset $S=\{b_{i_1}, b_{i_2}, \dots, b_{i_k}\} \in \mathscr{K}$
.
\end{claim}
To see this we first remark that $f(S)$ is a $k$-subset of $V(\pa{n}{h})$.
Furthermore, for each pair $b_{i_j},b_{i_{j+t}}\in S$, with $t>0$, we have
$$i_{j+t} + (j+t-1)h - (i_j + (j-1)h) = i_{j+t} - i_j + th > h\,.$$
Hence, by Definition \ref{def:h-path}, $(f(b_{i_j}),f(b_{i_{j+t}}))=(v_{i_j + (j-1)h},v_{i_{j+t} + (j+t-1)h})\notin E(\pa{n}{h})$. Thus,
$f(S)$ is an independent set of $\pa{n}{h}$.

\begin{claim}\label{claim:2} The map $f$ is bijective.
\end{claim}
It is easy to see that $f$ is injective.
Then, we consider the map $f^{-1}: \mathcal{I}_k \to \mathscr{K}$ such that, for any $S=\{v_{i_1},v_{i_2},\dots, v_{i_k}\} \in \mathcal{I}$, with $1\leq i_1<i_2<\cdots <i_k \leq n$,
\[
f^{-1}(\{v_{i_1},v_{i_2},\dots, v_{i_j}, \dots, v_{i_k}\}) = \{b_{i_1},b_{i_2 - h},\dots, b_{i_j - (j-1)h}, \dots, b_{i_k - (k-1)h}\}\,.
\]
Following the same steps as for $f$, one checks that $f^{-1}$
is injective. Thus, $f$ is surjective.

\medskip
We have established a bijection
between independent $k$-subsets of $\pa{n}{h}$ and $k$-subsets of a set with
$(n-hk+h) \geq 0$ elements. The lemma is proved.
\end{proof}

\noindent Some values of $p_{n,k}^{(h)}$ are shown in Tables \ref{tab:p_nk_1}--\ref{tab:p_nk_3} (Section \ref{sec:tables}).
The coefficients $p_{n,k}^{(h)}$ also enjoy the following property:
$p_{n,k}^{(h)} = p_{n-k+1,k}^{(h-1)}$.

\smallskip
For $n,h\geq0$, the number of independent sets of $\pa{n}{h}$ is
\[
p_n^{(h)}=\sum_{k\geq0}p_{n,k}^{(h)}=\sum_{k=0}^{\lceil n/(h+1)\rceil}p_{n,k}^{(h)}=\sum_{k=0}^{\lceil n/(h+1)\rceil}\binom{n-hk+h}{k}\,.
\]

\begin{remark} Denote by $F_n$ the $n^{th}$ element of the Fibonacci sequence: $F_1=1$, $F_2=1$, and $F_i=F_{i-1}+F_{i-2}$, for $i>2$. Then, $p_n^{(1)}=F_{n+2}$ is the number of vertices of the Fibonacci cube of order $n$.
\end{remark}


The following, simple fact is crucial for our work.

\begin{lemma}
\label{lem:p_n recurrence}
For $n, h \geq 0$,
\[
p_n^{(h)} = \begin{cases}
n+1 & \text{if }\ n \leq h+1\,, \\
p_{n-1}^{(h)}+p_{n-h-1}^{(h)} & \text{if }\ n>h+1\,.
\end{cases}
\]
\end{lemma}
A proof of this Lemma can also be obtained using the first part
of \cite[Proof of Theorem 1]{hoggatt}.
\begin{proof}
For $n \leq h+1$, by Definition \ref{def:indipendent subset},
the independent sets of $\pa{n}{h}$ have no more than $1$ element.
Thus, there are $n+1$ independent sets of $\pa{n}{h}$.

\smallskip
Consider the case $n > h+1$.
Let $\mathcal{I}$ be
the set of all independent sets of $\pa{n}{h}$, let
$\mathcal{I}_{in}$ be the set of the independent sets of $\pa{n}{h}$
that contain $v_n$, and let $\mathcal{I}_{out}=\mathcal{I}\setminus\mathcal{I}_{in}$.
The elements
of $\mathcal{I}_{out}$ are in one-to-one correspondence with the
$p_{n-1}^{(h)}$ independent sets of $\pa{n-1}{h}$, and
those of $\mathcal{I}_{in}$ are in one-to-one correspondence with the
$p_{n-h-1}^{(h)}$ independent sets of $\pa{n-h-1}{h}$.
\end{proof}

Tables \ref{tab:p_nk} displays a few values of $p_{n}^{(h)}$.

\section{Generalized Fibonacci numbers and generalized Fibonacci cubes}
\label{sec:hasse of paths}
Figure \ref{fig:P_3-4_h} shows a few Hasse diagrams $\hp{n}{h}$. Notice that,
as stated in the introduction, for each $n$, $\hp{n}{1}$ is the Fibonacci cube $\Gamma_n$.

\begin{figure}[h!]
  \centerline{\includegraphics[scale=0.6]{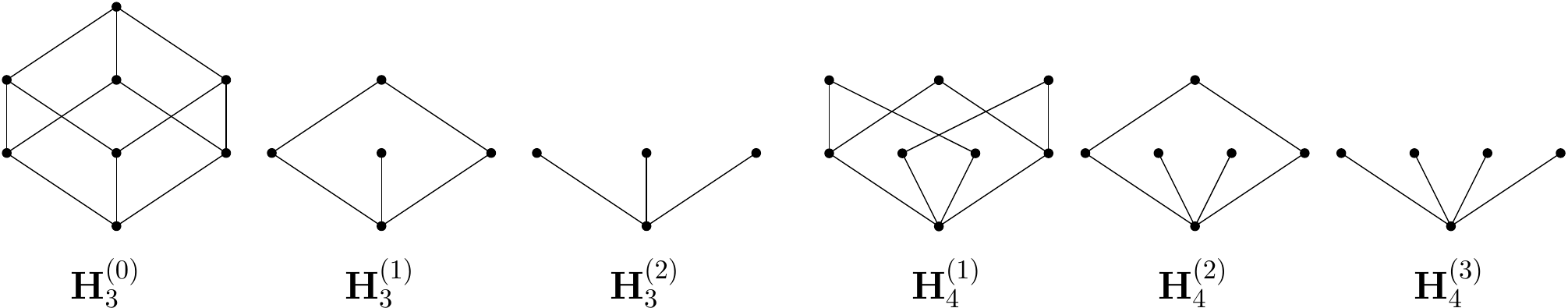}}
  \caption[Some $\hp{n}{h}$]
    {Some $\hp{n}{h}$.}
  \label{fig:P_3-4_h}
\end{figure}

Let $H_n^{(h)}$ be the number of edges of $\hp{n}{h}$.
Noting that in $\hp{n}{h}$ each non-empty independent $k$-subset covers
exactly $k$ independent $(k-1)$-subsets, we can write

\begin{equation}\label{eq:H_nh formula}
H_n^{(h)}= \sum_{k=1}^{\lceil n/(h+1)\rceil}kp_{n,k}^{(h)}= \sum_{k=1}^{\lceil n/(h+1)\rceil}k{{n-hk+h}\choose k}\ .
\end{equation}

\begin{remark}
$H_n^{(1)}$ counts the number of edges of $\Gamma_n$.
\end{remark}

Let now $T_{k,i}^{(n,h)}$ be the number of independent $k$-subsets of $\pa{n}{h}$ containing the vertex $v_i$,
and let, for $h,k\geq 0$, $n \in \mathbb{Z}$,
$\bar{p}_{n,k}^{(h)} = \begin{cases}
p_{0,k}^{(h)} & \text{if }\ n < 0\,, \\
p_{n,k}^{(h)} & \text{if }\ n \geq 0\,.
\end{cases}$

\begin{lemma}
\label{lem:T formula}
For $n,h,k \geq 0$, and $1\leq i \leq n$,
\[
T_{k,i}^{(n,h)}=\sum_{r=0}^{k-1}\, \bar{p}_{i-h-1,r}^{(h)} \ \bar{p}_{n-i-h,k-1-r}^{(h)}\,.
\]
\end{lemma}
\begin{proof}
No independent set of $\pa{n}{h}$ containing $v_i$ contains any of
the elements $v_{i-h}, \dots,v_{i-1},v_{i+1},\dots,v_{i+h}$.
Let $r$ and $s$ be non-negative integers whose sum is $k-1$. Each independent $k$-subset of $\pa{n}{h}$ containing $v_i$
can be obtained by adding $v_i$ to a $(k-1)$-subset $R\cup S$ such that

\noindent(a) $R\subseteq \{v_1,\dots,v_{i-h-1}\}$ is an independent $r$-subset of $\pa{n}{h}$;

\noindent(b) $S\subseteq \{v_{i+h+1},\dots,v_n\}$ is an independent $s$-subset of $\pa{n}{h}$.

\smallskip
Viceversa, one can obtain each of this pairs of subsets by removing $v_i$ from an independent
$k$-subset of $\pa{n}{h}$ containing $v_i$.
Thus, $T_{k,i}^{(n,h)}$ is obtained by counting independently
the subsets of type (a) and (b). The remark that the subsets of type (b) are
in bijection with the independent $s$-subsets of
$\pa{n-i-h}{h}$ proves the lemma.
\end{proof}

\begin{remark}
$T_{k,i}^{(n,1)}$ counts the number of strings $\alpha=b_1 b_2 \cdots b_n\in\Gamma_n$ such that:
(i) $H(\alpha,00\cdots0)=k$, and (ii) $b_i=1$.
\end{remark}

%

In order to obtain our main result, we prepare a lemma.

\begin{lemma}
\label{lem:somma tabelle T}
For positive $n$,
\[
\sum_{k=1}^{\lceil n/(h+1)\rceil}\sum_{i=1}^n T_{k,i}^{(n,h)} = H_n^{(h)}\,.
\]
\end{lemma}
\begin{proof}
The inner sum
counts the number of $k$-subsets exactly $k$ times, one for each element of the subset.
That is,
$\sum_{i=1}^n T_{k,i}^{(n,h)} = k p_{n,k}^{(h)}$.
Hence the lemma follows directly from Equation (\ref{eq:H_nh formula}).
\end{proof}

Next we introduce a family
of Fibonacci-like sequences.

\begin{definition}\label{def:gen fibonacci}
For $h\geq 0$, and $n\geq 1$, we define the \emph{$h$-Fibonacci sequence}
$\mathcal{F}^{(h)}=\{F_n^{(h)}\}_{n\geq 1}$ whose elements are
\[
F_n^{(h)} = \begin{cases}
1 & \text{if }\ n \leq h+1\,, \\
F_{n-1}^{(h)}+F_{n-h-1}^{(h)} & \text{if }\ n>h+1.
\end{cases}
\]
\end{definition}

The first values of the \emph{$h$-Fibonacci sequences}, for
$h\in\{1,\dots,10\}$, are shown in Table \ref{tab:F_n}.
From Lemma \ref{lem:p_n recurrence}, and setting
for $h\geq 0$, and $n \in \mathbb{Z}$,
$\bar{p}_{n}^{(h)} = \begin{cases}
p_{0}^{(h)} & \text{if }\ n < 0\,, \\
p_{n}^{(h)} & \text{if }\ n \geq 0\,,
\end{cases}$
we have that,
\begin{equation}\label{eq:F and p}
F_i^{(h)}=\bar{p}_{i-h-1}^{(h)}\,, \ \ \text{for each}\ \  i\geq 1\,.
\end{equation}

\smallskip
Thus, our Fibonacci-like sequences are obtained by prepending $h$
$1$'s to the sequence $p_{0}^{(h)},p_{1}^{(h)},\dots\ $. Therefore, we have:
\begin{itemize}
\item $\mathcal{F}^{(0)}=1,2,4,\dots,2^{n},\dots$;
\item $\mathcal{F}^{(1)}$ is the Fibonacci sequence;
\item more generally, $\mathcal{F}^{(h)}=\underbrace{1,\dots,1}_{h},p_{0}^{(h)},p_{1}^{(h)},p_{2}^{(h)},\dots$.
\end{itemize}

In the following, we define the discrete convolution operation $\ast$, as follows.
\begin{equation}\label{eq:convolution}
\left(\mathcal{F}^{(h)}\ast \mathcal{F}^{(h)}\right)(n)\doteq \sum_{i=1}^n F_{i}^{(h)} F_{n-i+1}^{(h)}
\end{equation}

\begin{theorem}
\label{th:main}
For $n,h\geq 0$, the following holds
\[
H_n^{(h)} = \left(\mathcal{F}^{(h)}\ast \mathcal{F}^{(h)}\right)(n)\,.
\]
\end{theorem}
\begin{proof}
The sum $\sum_{k=1}^{\lceil n/(h+1)\rceil} T_{k,i}^{(n,h)}$
counts the number of independent sets of $\pa{n}{k}$ containing $v_i$.
We can also obtain such a value by counting
the independent sets of both $\{v_1, \dots, v_{i-h-1}\}$, and $\{v_{i+h+1}, \dots, v_{n}\}$.
Thus,  we have:
\[
\sum_{k=1}^{\lceil n/(h+1)\rceil} T_{k,i}^{(n,h)} = \bar{p}_{i-h-1}^{(h)}\, \bar{p}_{n-h-i}^{(h)}\,.
\]
Using Lemma \ref{lem:somma tabelle T} we can write
\[
H_n^{(h)} = \sum_{k=1}^{\lceil n/(h+1)\rceil}\sum_{i=1}^n T_{k,i}^{(n,h)} =
\sum_{i=1}^n\sum_{k=1}^{\lceil n/(h+1)\rceil} T_{k,i}^{(n,h)} =
\sum_{i=1}^n \bar{p}_{i-h-1}^{(h)}\, \bar{p}_{n-h-i}^{(h)}.
\]
By Equation (\ref{eq:F and p}) we have
$\sum_{i=1}^n \bar{p}_{i-h-1}^{(h)}\, \bar{p}_{n-h-i}^{(h)}=\sum_{i=1}^n F_{i}^{(h)} F_{n-i+1}^{(h)}\,.$
By (\ref{eq:convolution}), the theorem is proved.
\end{proof}

We display some values of $H_n^{(h)}$ in Table \ref{tab:H_n} (Section \ref{sec:tables}).

\begin{remark} For $h=1$, we obtain the number of edges of $\Gamma_n$ by using Fibonacci numbers:
\[
H_n^{(1)} = \sum_{i=1}^n F_{i} F_{n-i+1}\,.
\]
The latter result is \cite[Proposition 3]{mediannature}.
\end{remark}

%
%
%

\section{The independent sets of powers of cycles}\label{sec:ind sub of cycles}

For $n, h, k\geq 0$, we denote by $c_{n,k}^{(h)}$ the number of  independent
$k$-subsets of $\cy{n}{h}$.
\begin{remark}
For $n>1$, $c_{n,k}^{(1)}$ counts the number of binary strings $\alpha\in \Lambda_n$
with $k$ $1$'s.
\end{remark}
\begin{lemma}\label{lem:q_nk formula}
For $n, h\geq 0$, and $k>1$,
\[
c_{n,k}^{(h)} = \frac{n}{k}\binom{n-hk-1}{k-1}\ .
\]
Moreover, $c_{n,0}^{(h)}=1$, and $c_{n,1}^{(h)}=n$, for each $n, h\geq 0$.
\end{lemma}
\begin{proof}
Fix an element $v_i \in V(\cy{n}{h})$, and let $n>2h$.
Any independent set of $\cy{n}{h}$
containing $v_i$ does not contain the $h$ elements preceding  $v_i$ and the $h$
elements following $v_i$. Thus, the number of
independent $k$-subsets of $\cy{n}{h}$ containing $v_i$ equals
\[
p_{n-2h-1,k-1}^{(h)}=\binom{n-hk-1}{k-1}\ .
\]
The total number of independent $k$-subsets of $\cy{n}{h}$ is obtained
by multiplying $p_{n-2h-1,k-1}^{(h)}$ by $n$, then dividing it by $k$ (each
subset is counted $k$ times by the previous proceeding).
The case $n\leq 2h$, as well as the cases $k=0,1$, can be easily verified.
\end{proof}

\noindent Some values of $c_{n,k}^{(h)}$ are displayed in Tables \ref{tab:q_nk_1}--\ref{tab:q_nk_3}.

For $n,h\geq0$, the number of  all independent sets of $\cy{n}{h}$ is
\begin{equation}\label{eq:q_n def}
c_n^{(h)}=\sum_{k\geq0}c_{n,k}^{(h)}=\sum_{k=0}^{\lceil n/(h+1)\rceil}c_{n,k}^{(h)}\ ,
\end{equation}
\begin{remark} Denote by $L_n$ the $n^{th}$ element of the Lucas sequence $L_1=1$, $L_2=3$, and $L_i=L_{i-1}+L_{i-2}$, for $i>2$. Then, for $n>1$, $c_n^{(1)}=L_{n}$ is the number of elements of the Lucas cube of order $n$.
\end{remark}

Some values of $c_{n}^{(h)}$ are shown in Table \ref{tab:q_nk}.
The coefficients $c_n^{(h)}$ satisfy a recursion that closely resembles that of Lemma 2.2.
\begin{lemma}
\label{lem:q_n recurrence}
For $n, h \geq 0$,
\begin{equation}\label{eq:q_n recurrence}
c_n^{(h)} = \begin{cases}
n+1 & \text{if }\ n\leq 2h+1\,, \\
c_{n-1}^{(h)}+c_{n-h-1}^{(h)} & \text{if }\ n>2h+1.
\end{cases}
\end{equation}
\end{lemma}
\begin{proof}
The case $n\leq 2h+1$ can be easily checked. The case $2h+1< n \leq 3h+2$
is discussed at the end of this proof.
Let $n> 3h+2$, and let $\mathcal{I}$ be the set of the independent sets of $\cy{n}{h}$. Let $\mathcal{I}_{in}$ be
the subset of these sets that (i) do not contain $v_n$, and that  (ii) do not contain
any of the following pairs: $(v_1, v_{n-h}), (v_2, v_{n-h+1}), \dots, (v_h, v_{n-1})$. Let
then $\mathcal{I}_{out}$ be the subset of the remaining independent sets of $\cy{n}{h}$.

\smallskip
It is easy to see that the elements of $\mathcal{I}_{in}$ are exactly the independent sets
of $\cy{n-1}{h}$. Indeed, $v_n$ is not a vertex of $\cy{n-1}{h}$ and the vertices
of pairs $(v_1, v_{n-h})$, $(v_2, v_{n-h+1})$, $\dots$, $(v_h, v_{n-1})$ are adjacent in $\cy{n-1}{h}$.
On the other hand, to show that
\[
|\mathcal{I}_{out}|=c_{n-h-1}^{(h)}
\]
we argue as follows. First we recall (see the proof of Lemma \ref{lem:q_nk formula})
that the number of independent $k$-subsets of
$\cy{n}{h}$ that contain $v_n$ is $p_{n-2h-1,k-1}^{(h)}$.
Secondly we claim that the number of independent $k$-subsets
of $\cy{n}{h}$ containing one of the pairs $(v_1, v_{n-h})$, $(v_2, v_{n-h+1})$, $\dots$,
$(v_h, v_{n-1})$ is $h p_{n-3h-2,k-2}^{(h)}$. To see this, consider the pair $(v_1,v_{n-h})$. The independent
sets containing such a pair do not contain the $h$ vertices from $v_n-h+1$ to $v_n$, do
not contain the $h$ vertices from $v_2$ to $v_{h+1}$, and do not contain the $h$ vertices from
$v_{n-2h}$ to $v_{n-h-1}$. Thus, the removal of such vertices and of the vertices $v_1$ and $v_{n-h}$
turns $\cy{n}{h}$ into $\pa{n-3h-2}{h}$. Hence we can obtain all the independent $k$-subsets of
$\cy{n}{h}$ that contain the pair $(v_1,v_{n-h})$ by simply adding these two vertices to one of the
$p_{n-3h-2,k-2}^{(h)}$ independent $k-2$-subsets of $\pa{n-3h-2}{h}$. Same reasoning
can be carried out for any other one of the pairs: $(v_2, v_{n-h+1})$, $\dots$,
$(v_h, v_{n-1})$.

\smallskip
Using Lemmas \ref{lem:p_nk formula} and \ref{lem:q_nk formula} one can easily derive
that
\[
p_{n-2h-1,k-1}^{(h)}+hp_{n-3h-2,k-2}^{(h)}=c_{n-h-1,k-1}^{(h)}\,.
\]
Hence, we derive the size of $\mathcal{I}_{out}$:
\[
|\mathcal{I}_{out}|=c_{n-h-1}^{(h)}=\sum_{k\geq 1}p_{n-2h-1,k-1}^{(h)}\ + h\sum_{k\geq2}p_{n-3h-2,k-2}^{(h)}\,.
\]
Summing up we have shown that $|\mathcal{I}|=|\mathcal{I}_{in}|+|\mathcal{I}_{out}|$, that is
\[
c_n^{(h)}=c_{n-1}^{(h)}+c_{n-h-1}^{(h)}\,.
\]

The proof of the case $2h+1< n \leq 3h+2$
is obtained in a similar way, observing that $|\mathcal{I}_{out}|=n-h$, and that
$n-h-1\leq 2h+1$.
%
%
%
%
%
%
%
%
%
%
%
%
%
%
%
\end{proof}


\section{Generalized Lucas cubes and Lucas numbers}
\label{sec:hasse of cycles}

Figure \ref{fig:Q_3-4_h} shows a few Hasse diagrams $\hc{n}{h}$. Notice that,
as stated in the introduction, for each $n\geq 1$, $\hc{n}{1}$ is the Lucas cube $\Lambda_n$.

\begin{figure}[h!]
  \centerline{\includegraphics[scale=0.6]{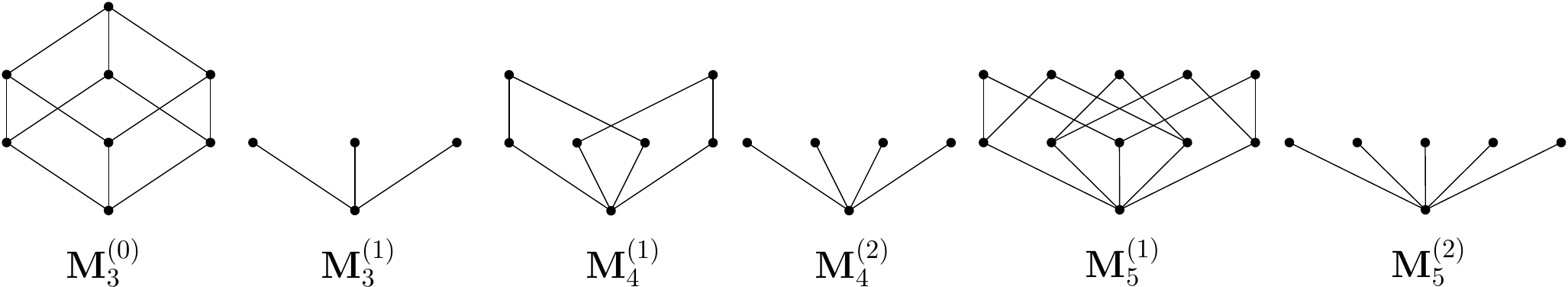}}
  \caption[Some $\hc{n}{h}$]
    {Some $\hc{n}{h}$.}
  \label{fig:Q_3-4_h}
\end{figure}

Let $M_n^{(h)}$ be the number of edges of $\hc{n}{h}$.
As done in Section \ref{sec:hasse of paths} for the case of paths, we immediately provide
a formula for $M_n^{(h)}$:
\begin{equation}\label{eq:M_nh formula}
M_n^{(h)}= \sum_{k=0}^{\lceil n/h+1 \rceil}kc_{n,k}^{(h)} = n \sum_{k=0}^{\lceil n/h+1 \rceil}{\binom{n-hk-1}{k-1}}\ .
\end{equation}
\begin{remark}
For $n>1$, $M_n^{(1)}$ counts the number of edges of $\Lambda_n$.
As shown in \cite[Proposition 4(ii)]{lucas}, $M_n^{(1)}=n F_{n-1}$.
\end{remark}
\smallskip
As shown in the proof of Lemma \ref{lem:q_nk formula}, the value
\[
p_{n-2h-1,k-1}^{(h)}=\binom{n-hk-1}{k-1}
\]
is the analogue of the coefficient $T_{k,i}^{(n,h)}$:  in the case of cycles we have no dependencies on $i$,
because each choice of vertex is equivalent.
We can obtain $M_n^{(h)}$ in terms of a Fibonacci-like sequence, as follows.
\begin{proposition}\label{prop:M_n}
For $n>h\geq 0$, the following holds
\[
M_n^{(h)} = n F_{n-h}^{(h)}\,.
\]
\end{proposition}
\begin{proof}
Using Equation (\ref{eq:F and p}) we obtain:
\[
M_n^{(h)} = n\sum_{k=1}^{\lceil n/(h+1)\rceil} \bar{p}_{n-2h-1,k-1}^{(h)} =
n \bar{p}_{n-2h-1}^{(h)}= n F_{n-h}^{(h)}\,.
\]
\end{proof}

In analogy with Section \ref{sec:hasse of paths}, we introduce
a family of Lucas-like sequences.

\begin{definition}\label{def:gen lucas}
For $h\geq 0$, and $n\geq 1$, we define the \emph{$h$-Lucas sequence}
$\mathcal{L}^{(h)}=\{L_n^{(h)}\}_{n\geq 1}$ whose elements are
\[
L_n^{(h)} = \begin{cases}
h+1 & \text{if }\ n = 1\,, \\
1 & \text{if }\ 2 \leq n \leq h+1\,, \\
L_{n-1}^{(h)}+L_{n-h-1}^{(h)} & \text{if }\ n>h+1.
\end{cases}
\]
\end{definition}

The first values of the $h$-Lucas sequences, for $h\in\{0,\dots,10\}$,
are displayed in Table \ref{tab:L_n}.
We have that,
\begin{equation}\label{eq:L and q}
L_i^{(h)}={c}_{i-1}^{(h)}\,, \ \ \text{for each}\ \  i> h+1\,.
\end{equation}

To prove the main result of this section, the following lemma is needed.

\begin{lemma}\label{lem:lucas to fibo}
For $n>h\geq 0$, the following holds
\[
L_{n+1}^{(h)}=F_{n}^{(h)}+(h+1)F_{n-h}^{(h)}
\]
\end{lemma}
\begin{proof}
The result is proved by induction on $n$. Indeed,
$
L_{n+1}^{(h)}=L_{n}^{(h)}+L_{n-h}^{(h)}\,.
$
Applying the inductive hypothesis, we have
\[
L_{n+1}^{(h)}=F_{n-1}^{(h)}+(h+1)F_{n-h-1}^{(h)}+F_{n-h-1}^{(h)}+(h+1)F_{n-2h-1}^{(h)}=F_{n}^{(h)}+(h+1)F_{n-h}^{(h)}\,.
\]

\end{proof}

Finally, our analogous of Theorem \ref{th:main}, for cycles, is the following.

\begin{theorem}
\label{th:main_cycles}
For $n>h\geq 0$, the following holds
\[
M_n^{(h)} = \left(\mathcal{F}^{(h)}\ast \mathcal{L}^{(h)}\right)(n-h)\,.
\]
\end{theorem}
\begin{proof}
By Proposition \ref{prop:M_n}, the statement of the Theorem is equivalent to
\begin{equation}\label{eq:main2}
\sum_{i=1}^{n-h}F_{i}^{(h)}L_{n-h+1-i}^{(h)}=n F_{n-h}^{(h)}\,.
\end{equation}

Let $h=0$. We have
\begin{equation*}
\sum_{i=1}^{n}F_{i}^{(0)}L_{n+1-i}^{(0)}=
\sum_{i=1}^{n}2^{i-1}2^{n-i}=\sum_{i=1}^{n}2^{n-1}=
n 2^{n-1}=n F_{n}^{(0)}\,.
\end{equation*}

Let $h=1$. In this case the statement of the theorem reduces to the well known identity involving (classical) Fibonacci and Lucas sequences:
\begin{equation*}
\sum_{i=1}^{n}F_{i}L_{n-i+1}=(n+1) F_{n}\,.
\end{equation*}

Let $h\geq 2$. We prove \eqref{eq:main2} by induction on $n$. If $n=h+1$, then
\begin{equation*}
\sum_{i=1}^{n-h}F_{i}^{(h)}L_{n-h+1-i}^{(h)}=F_{1}^{(h)}L_{1}^{(h)}=h+1=n F_{1}^{(h)}\,.
\end{equation*}
Let $\bar{n}>h+1$, and suppose (inductive hypothesis) that \eqref{eq:main2} holds
for every $1<n\leq\bar{n}$. Let $m=\bar{n}-h$ (and note that $m\geq 2$). We need to prove that
\begin{equation}\label{eq:m2induction}
\sum_{i=1}^{m+1}F_{i}^{(h)}L_{m+2-i}^{(h)}=(m+h+1) F_{m+1}^{(h)}\,.
\end{equation}
We find it convenient to define, for $h\geq 2$, the following integer sequences, which extend $F_n^{(h)}$ and $L_n^{(h)}$
to a range of negative integers.
\begin{align}
\label{eq:fiboneg}&\bar{F}_n^{(h)} = \begin{cases}
1 & \text{if }\ n = -h\,, \\
0 & \text{if }\ -h < n \leq 0\,, \\
\bar{F}_{n-1}^{(h)}+\bar{F}_{n-h-1}^{(h)} & \text{if }\ n>0.
\end{cases}
\\
\label{eq:lucasneg}&\bar{L}_n^{(h)} = \begin{cases}
h+1 & \text{if }\ n = -h\,, \\
-h & \text{if }\ n = -h+1\,, \\
0 & \text{if }\ -h+1 < n \leq 0\,, \\
\bar{L}_{n-1}^{(h)}+\bar{L}_{n-h-1}^{(h)} & \text{if }\ n>0.
\end{cases}
\end{align}
Using Definitions \ref{def:gen fibonacci} and \ref{def:gen lucas} one can easily check that,
for $n>0$, ${F}_n^{(h)}=\bar{F}_n^{(h)}$, and
${L}_n^{(h)}=\bar{L}_n^{(h)}$. Hence, applying the recurrences in \eqref{eq:fiboneg} and
\eqref{eq:lucasneg}, we obtain
\begin{align}
\notag\sum_{i=1}^{m+1}F_{i}^{(h)}L_{m+2-i}^{(h)}&=\sum_{i=1}^{m+1}\bar{F}_{i}^{(h)}\bar{L}_{m+2-i}^{(h)}=\\
\label{eq:D1}
&=\sum_{i=1}^{m+1}\bar{F}_{i-1}^{(h)}\bar{L}_{m+1-i}^{(h)}\,+\\
\label{eq:D2}
&+\sum_{i=1}^{m+1}\bar{F}_{i-1}^{(h)}\bar{L}_{m+1-h-i}^{(h)}\,+\\
\label{eq:D3}
&+\sum_{i=1}^{m+1}\bar{F}_{i-h-1}^{(h)}\bar{L}_{m+1-i}^{(h)}\,+\\
\label{eq:D4}
&+\sum_{i=1}^{m+1}\bar{F}_{i-h-1}^{(h)}\bar{L}_{m+1-h-i}^{(h)}\,.
\end{align}

We compute the sums \eqref{eq:D1}--\eqref{eq:D4} separately.

\medskip
Direct computation shows that \eqref{eq:D1} equals
\begin{equation}\label{eq:D1b}
\sum_{i=1}^{m-1}F_{i}^{(h)}L_{m-i}^{(h)}+
\bar{F}_{0}^{(h)}\bar{L}_{m}^{(h)}+
\bar{F}_{m}^{(h)}\bar{L}_{0}^{(h)}\,.
\end{equation}
Applying the inductive hypotheses to the first term of \eqref{eq:D1b}, and observing that
$\bar{L}_{0}^{(h)}=\bar{F}_{0}^{(h)}=0$, we obtain that the sum \eqref{eq:D1} is
\begin{equation}\label{eq:D1c}
(m+h-1)F_{m-1}^{(h)}\,.
\end{equation}

\medskip
In order to compute \eqref{eq:D2} we distinguish two cases. If $m> h+1$, direct computation shows that \eqref{eq:D2} equals
\begin{equation}\label{eq:D2b}
\sum_{i=1}^{m-h-1}F_{i}^{(h)}L_{m-h-i}^{(h)}+
\bar{F}_{0}^{(h)}\bar{L}_{m-h}^{(h)}+
\bar{F}_{m-h}^{(h)}\bar{L}_{0}^{(h)}+
\bar{F}_{m-h+1}^{(h)}\bar{L}_{-1}^{(h)}+\cdots+
\bar{F}_{m}^{(h)}\bar{L}_{-h}^{(h)}.
\end{equation}
We note that $m > h+1 \geq 3$. Applying the inductive hypotheses to the first term of \eqref{eq:D2b}, and observing that
all other terms are zeroes, except $\bar{F}_{m}^{(h)}\bar{L}_{-h}^{(h)}=(h+1){F}_{m}^{(h)}$,
and $\bar{F}_{m-1}^{(h)}\bar{L}_{-h+1}^{(h)}=-h {F}_{m-1}^{(h)}$, we obtain that \eqref{eq:D2}
equals
\begin{flalign}\label{eq:D2c1}
&(m-1)F_{m-h-1}^{(h)}+(h+1) F_{m}^{(h)} - h F_{m-1}^{(h)}\,.\ \ \ &(\text{for}\ \ m> h+1)
\end{flalign}

If, on the other hand, $m\leq h+1$, we observe that the summands of \eqref{eq:D2}
vanish, with the exception of  $\bar{F}_{m}^{(h)}\bar{L}_{-h}^{(h)}=(h+1)F_m^{(h)}$,
and $\bar{F}_{m-1}^{(h)}\bar{L}_{-h+1}^{(h)}=-h F_{m-1}^{(h)}$.
Indeed, $m\geq 2$. Thus, the sum \eqref{eq:D2} is
\begin{flalign}\label{eq:D2c2}
&(h+1) F_{m}^{(h)} - h F_{m-1}^{(h)}\,.\ \ \ &(\text{for}\ \ m\leq h+1)
\end{flalign}

\medskip
A similar argument shows that \eqref{eq:D3} equals
\begin{flalign}
\label{eq:D3c1}&(m-1)F_{m-h-1}^{(h)}+L_{m}^{(h)}\,,\ \ \ &(\text{for}\ \ m> h+1)\\
\label{eq:D3c2}&L_{m}^{(h)}\,.\ \ \ &(\text{for}\ \ m\leq h+1)
\end{flalign}

\medskip
To calculate \eqref{eq:D4} we distinguish the four cases $m\leq h$, $m=h+1$, $h+1<m\leq 2h+1$,
and $m> 2h+1$. If $m> 2h+1$, \eqref{eq:D4} equals
\begin{equation*}
\sum_{i=1}^{m-2h-1}F_{i}^{(h)}L_{m-2h-i}^{(h)}+
\bar{L}_{-h}^{(h)}\bar{F}_{m-h}^{(h)}+
\cdots+\bar{L}_{0}^{(h)}\bar{F}_{m-2h}^{(h)}+
\bar{F}_{-h}^{(h)}\bar{L}_{m-h}^{(h)}+\cdots+
\bar{F}_{0}^{(h)}\bar{L}_{m-2h}^{(h)}.
\end{equation*}
We note that $m-h-1>h>0$. Applying the inductive hypotheses to the first term of the preceding sum,
and observing that all other terms are zeroes, except
$\bar{L}_{-h}^{(h)}\bar{F}_{m-h}^{(h)}=(h+1)F_{m-h}^{(h)}$,
$\bar{L}_{-h+1}^{(h)}\bar{F}_{m-h-1}^{(h)}=-h F_{m-h-1}^{(h)}$,
and $\bar{F}_{-h}^{(h)}\bar{L}_{m-h}^{(h)}=L_{m-h}^{(h)}$,
we obtain that \eqref{eq:D4} is
\begin{flalign}\label{eq:D4c1}
&(m-h-1)F_{m-2h-1}^{(h)}+(h+1)F_{m-h}^{(h)}-hF_{m-h-1}^{(h)}+L_{m-h}^{(h)}.&(\text{for}\ \ m> 2h+1)
\end{flalign}
To tackle the other cases, we expand the sum \eqref{eq:D4}, and we check how
the non-zeroes terms change depending on the constraints on m. We obtain that
\eqref{eq:D4} equals
\begin{flalign}
\label{eq:D4c2}&(h+1)F_{m-h}^{(h)}-hF_{m-h-1}^{(h)}+L_{m-h}^{(h)}\,, &(\text{for}\ \ h+1<m\leq 2h+1)\\
\label{eq:D4c3}&(h+1)F_{m-h}^{(h)}+L_{m-h}^{(h)}\,, &(\text{for}\ \ m=h+1)\\
\label{eq:D4c4}&0\,. &(\text{for}\ \ m\leq h)
\end{flalign}

\medskip
Finally, we add up the four summands \eqref{eq:D1}--\eqref{eq:D4}, distinguishing the four identified cases.  If $m> 2h+1$, we sum \eqref{eq:D1c}, \eqref{eq:D2c1}, \eqref{eq:D3c1}, and \eqref{eq:D4c1},
obtaining
\begin{multline*}
\sum_{i=1}^{m+1} F_{i}^{(h)}L_{m+1-i}^{(h)}=(m+h-1)F_{m-1}^{(h)}
+2(m-1)F_{m-h-1}^{(h)}+(h+1) F_{m}^{(h)} - h F_{m-1}^{(h)}+\\
+L_{m}^{(h)}+(m-h-1)F_{m-2h-1}^{(h)}+(h+1)F_{m-h}^{(h)}-hF_{m-h-1}^{(h)}+L_{m-h}^{(h)}\,.
\end{multline*}
Repeatedly applying the recurrences $F_i^{(h)}=F_{i-1}^{(h)}+F_{i-h-1}^{(h)}$ and
$L_i^{(h)}=L_{i-1}^{(h)}+L_{i-h-1}^{(h)}$, and using Lemma \ref{lem:lucas to fibo}, we obtain
\begin{multline*}
\sum_{i=1}^{m+1}F_{i}^{(h)}L_{m+2-i}^{(h)}=
(m+h+1)F_{m+1}^{(h)}\,,\\
\end{multline*}
which proves \eqref{eq:main2} in the case $m> 2h+1$.

\medskip
If $h+1<m\leq 2h+1$, we sum \eqref{eq:D1c}, \eqref{eq:D2c1}, \eqref{eq:D3c1}, and \eqref{eq:D4c2},
and we obtain
\begin{multline*}
\sum_{i=1}^{m+1} F_{i}^{(h)}L_{m+2-i}^{(h)}=(m+h-1)F_{m-1}^{(h)}
+2(m-1)F_{m-h-1}^{(h)}+(h+1) F_{m}^{(h)} - h F_{m-1}^{(h)}+\\
+L_{m}^{(h)}+(h+1)F_{m-h}^{(h)}-hF_{m-h-1}^{(h)}+L_{m-h}^{(h)}\,.
\end{multline*}
Note that, by Definition \ref{def:gen fibonacci}, we have $F_{h+k}^{(h)}=k$, for $k\in\{1,h+2\}$. Moreover,
$F_{m-h-1}^{(h)}=F_{m-h}^{(h)}=1$, and, by Definition \ref{def:gen lucas}, $L_{m-h}^{(h)}=1$.
Applying Lemma \ref{lem:lucas to fibo} to $L_{m}^{(h)}$,  we obtain
\begin{align*}
\sum_{i=1}^{m+1} F_{i}^{(h)}L_{m+2-i}^{(h)}&
=(h+1)F_m^{(h)}+mF_{m-1}^{(h)}+h+1+2m=\\&
=(h+1)(m-h)+m(m-1-h)+h+1+2m=\\&
=(m+h+1)(m+1-h)=(m+h+1)F_{m+1}^{(h)}\,.
\end{align*}

\medskip
If $m=h+1$, we sum \eqref{eq:D1c}, \eqref{eq:D2c2}, \eqref{eq:D3c2}, and \eqref{eq:D4c3},
obtaining
\begin{multline*}
\sum_{i=1}^{m+1} F_{i}^{(h)}L_{m+2-i}^{(h)}=(m+h-1)F_{m-1}^{(h)}
+(h+1) F_{m}^{(h)} - h F_{m-1}^{(h)}+\\
+L_{m}^{(h)}+(h+1)F_{m-h}^{(h)}+L_{m-h}^{(h)}\,.
\end{multline*}
Noting that, by Definition \ref{def:gen fibonacci} and Definition \ref{def:gen lucas},
$F_{m}^{(h)}=F_{m-1}^{(h)}=\dots=F_{1}^{(h)}=1$, $L_{m}^{(h)}=1$, and $L_{m-h}^{(h)}=h+1$, we have
\begin{align*}
\sum_{i=1}^{m+1} F_{i}^{(h)}L_{m+2-i}^{(h)}=
4h+4=(m+h+1)2=(m+h+1)F_{m+1}^{(h)}\,.
\end{align*}

\medskip
In the last case, if $m<h+1$, we sum \eqref{eq:D1c}, \eqref{eq:D2c2}, \eqref{eq:D3c2}, and \eqref{eq:D4c4},
and we have
\begin{multline*}
\sum_{i=1}^{m+1} F_{i}^{(h)}L_{m+2-i}^{(h)}=(m+h-1)F_{m-1}^{(h)}
+(h+1) F_{m}^{(h)} - h F_{m-1}^{(h)}+L_{m}^{(h)}\,.
\end{multline*}
Proceeding as in the previous case, and recalling that $m\geq 2$, we obtain
\begin{align*}
\sum_{i=1}^{m+1} F_{i}^{(h)}L_{m+2-i}^{(h)}=
m+h-1+h+1-h+1=
(m+h+1)F_{m+1}^{(h)}\,.
\end{align*}

\medskip
In all cases we obtain the desired result, and the proof is complete.
\end{proof}

We display some values of $M_n^{(h)}$ in Table \ref{tab:M_n}.

\section{Tables}
\label{sec:tables}

We collect here some values obtained by computing the formula presented in the preceding sections.

\begin{center}
\noindent\(\footnotesize{\arraycolsep=2.2pt
\begin{array}{c|cccccccccccccccc}
  & \text{n=0} & 1 & 2 & 3 & 4 & 5 & 6 & 7 & 8 & 9 & 10 & 11 & 12 & 13 & 14 & 15 \\
\hline
 \text{k=0} & 1 & 1 & 1 & 1 & 1 & 1 & 1 & 1 & 1 & 1 & 1 & 1 & 1 & 1 & 1 & 1 \\
 1 & 0 & 1 & 2 & 3 & 4 & 5 & 6 & 7 & 8 & 9 & 10 & 11 & 12 & 13 & 14 & 15 \\
 2 & 0 & 0 & 0 & 1 & 3 & 6 & 10 & 15 & 21 & 28 & 36 & 45 & 55 & 66 & 78 & 91 \\
 3 & 0 & 0 & 0 & 0 & 0 & 1 & 4 & 10 & 20 & 35 & 56 & 84 & 120 & 165 & 220 & 286 \\
 4 & 0 & 0 & 0 & 0 & 0 & 0 & 0 & 1 & 5 & 15 & 35 & 70 & 126 & 210 & 330 & 495 \\
 5 & 0 & 0 & 0 & 0 & 0 & 0 & 0 & 0 & 0 & 1 & 6 & 21 & 56 & 126 & 252 & 462 \\
 6 & 0 & 0 & 0 & 0 & 0 & 0 & 0 & 0 & 0 & 0 & 0 & 1 & 7 & 28 & 84 & 210 \\
 7 & 0 & 0 & 0 & 0 & 0 & 0 & 0 & 0 & 0 & 0 & 0 & 0 & 0 & 1 & 8 & 36 \\
 8 & 0 & 0 & 0 & 0 & 0 & 0 & 0 & 0 & 0 & 0 & 0 & 0 & 0 & 0 & 0 & 1 \\
\end{array}
}\)
\captionof{table}{The number $p_{n,k}^{(1)}$ of  independent $k$-subsets of $\pa{n}{1}$}
\label{tab:p_nk_1}
\end{center}

\begin{center}
\noindent\(\footnotesize{\arraycolsep=2.2pt
\begin{array}{c|ccccccccccccccccc}
  & \text{n=0} & 1 & 2 & 3 & 4 & 5 & 6 & 7 & 8 & 9 & 10 & 11 & 12 & 13 & 14 & 15 & 16 \\
\hline
 \text{k=0} & 1 & 1 & 1 & 1 & 1 & 1 & 1 & 1 & 1 & 1 & 1 & 1 & 1 & 1 & 1 & 1 & 1 \\
 1 & 0 & 1 & 2 & 3 & 4 & 5 & 6 & 7 & 8 & 9 & 10 & 11 & 12 & 13 & 14 & 15 & 16 \\
 2 & 0 & 0 & 0 & 0 & 1 & 3 & 6 & 10 & 15 & 21 & 28 & 36 & 45 & 55 & 66 & 78 & 91 \\
 3 & 0 & 0 & 0 & 0 & 0 & 0 & 0 & 1 & 4 & 10 & 20 & 35 & 56 & 84 & 120 & 165 & 220 \\
 4 & 0 & 0 & 0 & 0 & 0 & 0 & 0 & 0 & 0 & 0 & 1 & 5 & 15 & 35 & 70 & 126 & 210 \\
 5 & 0 & 0 & 0 & 0 & 0 & 0 & 0 & 0 & 0 & 0 & 0 & 0 & 0 & 1 & 6 & 21 & 56 \\
 6 & 0 & 0 & 0 & 0 & 0 & 0 & 0 & 0 & 0 & 0 & 0 & 0 & 0 & 0 & 0 & 0 & 1 \\
\end{array}
}\)
\captionof{table}{The number $p_{n,k}^{(2)}$ of  independent $k$-subsets of $\pa{n}{2}$}
\label{tab:p_nk_2}
\end{center}

\begin{center}
\noindent\(\footnotesize{\arraycolsep=2.2pt
\begin{array}{c|cccccccccccccccccc}
  & \text{n=0} & 1 & 2 & 3 & 4 & 5 & 6 & 7 & 8 & 9 & 10 & 11 & 12 & 13 & 14 & 15 & 16 & 17 \\
\hline
 \text{k=0} & 1 & 1 & 1 & 1 & 1 & 1 & 1 & 1 & 1 & 1 & 1 & 1 & 1 & 1 & 1 & 1 & 1 & 1 \\
 1 & 0 & 1 & 2 & 3 & 4 & 5 & 6 & 7 & 8 & 9 & 10 & 11 & 12 & 13 & 14 & 15 & 16 & 17 \\
 2 & 0 & 0 & 0 & 0 & 0 & 1 & 3 & 6 & 10 & 15 & 21 & 28 & 36 & 45 & 55 & 66 & 78 & 91 \\
 3 & 0 & 0 & 0 & 0 & 0 & 0 & 0 & 0 & 0 & 1 & 4 & 10 & 20 & 35 & 56 & 84 & 120 & 165 \\
 4 & 0 & 0 & 0 & 0 & 0 & 0 & 0 & 0 & 0 & 0 & 0 & 0 & 0 & 1 & 5 & 15 & 35 & 70 \\
 5 & 0 & 0 & 0 & 0 & 0 & 0 & 0 & 0 & 0 & 0 & 0 & 0 & 0 & 0 & 0 & 0 & 0 & 1 \\
\end{array}
}\)
\captionof{table}{The number $p_{n,k}^{(3)}$ of  independent $k$-subsets of $\pa{n}{3}$}
\label{tab:p_nk_3}
\end{center}
%

\begin{center}
\noindent\(\footnotesize{\arraycolsep=2.2pt
\begin{array}{c|cccccccccccccc}
  & \text{n=0} & 1 & 2 & 3 & 4 & 5 & 6 & 7 & 8 & 9 & 10 & 11 & 12 & 13 \\
\hline
 \text{h=0} & 1 & 2 & 4 & 8 & 16 & 32 & 64 & 128 & 256 & 512 & 1024 & 2048 & 4096 & 8192 \\
 1 & 1 & 2 & 3 & 5 & 8 & 13 & 21 & 34 & 55 & 89 & 144 & 233 & 377 & 610 \\
 2 & 1 & 2 & 3 & 4 & 6 & 9 & 13 & 19 & 28 & 41 & 60 & 88 & 129 & 189 \\
 3 & 1 & 2 & 3 & 4 & 5 & 7 & 10 & 14 & 19 & 26 & 36 & 50 & 69 & 95 \\
 4 & 1 & 2 & 3 & 4 & 5 & 6 & 8 & 11 & 15 & 20 & 26 & 34 & 45 & 60 \\
 5 & 1 & 2 & 3 & 4 & 5 & 6 & 7 & 9 & 12 & 16 & 21 & 27 & 34 & 43 \\
 6 & 1 & 2 & 3 & 4 & 5 & 6 & 7 & 8 & 10 & 13 & 17 & 22 & 28 & 35 \\
 7 & 1 & 2 & 3 & 4 & 5 & 6 & 7 & 8 & 9 & 11 & 14 & 18 & 23 & 29 \\
 8 & 1 & 2 & 3 & 4 & 5 & 6 & 7 & 8 & 9 & 10 & 12 & 15 & 19 & 24 \\
 9 & 1 & 2 & 3 & 4 & 5 & 6 & 7 & 8 & 9 & 10 & 11 & 13 & 16 & 20 \\
 10 & 1 & 2 & 3 & 4 & 5 & 6 & 7 & 8 & 9 & 10 & 11 & 12 & 14 & 17 \\
\end{array}
}\)
\captionof{table}{The number $p_{n}^{(h)}$ of  all independent sets of $\pa{n}{h}$}
\label{tab:p_nk}
\end{center}

\begin{center}
\noindent\(\footnotesize{\arraycolsep=2.2pt
\begin{array}{c|ccccccccccccccc}
  & \text{n=1} & 2 & 3 & 4 & 5 & 6 & 7 & 8 & 9 & 10 & 11 & 12 & 13 & 14 & 15 \\
\hline
 \text{h=0} & 1 & 2 & 4 & 8 & 16 & 32 & 64 & 128 & 256 & 512 & 1024 & 2048 & 4096 & 8192 & 16384 \\
 1 & 1 & 1 & 2 & 3 & 5 & 8 & 13 & 21 & 34 & 55 & 89 & 144 & 233 & 377 & 610 \\
 2 & 1 & 1 & 1 & 2 & 3 & 4 & 6 & 9 & 13 & 19 & 28 & 41 & 60 & 88 & 129 \\
 3 & 1 & 1 & 1 & 1 & 2 & 3 & 4 & 5 & 7 & 10 & 14 & 19 & 26 & 36 & 50 \\
 4 & 1 & 1 & 1 & 1 & 1 & 2 & 3 & 4 & 5 & 6 & 8 & 11 & 15 & 20 & 26 \\
 5 & 1 & 1 & 1 & 1 & 1 & 1 & 2 & 3 & 4 & 5 & 6 & 7 & 9 & 12 & 16 \\
 6 & 1 & 1 & 1 & 1 & 1 & 1 & 1 & 2 & 3 & 4 & 5 & 6 & 7 & 8 & 10 \\
 7 & 1 & 1 & 1 & 1 & 1 & 1 & 1 & 1 & 2 & 3 & 4 & 5 & 6 & 7 & 8 \\
 8 & 1 & 1 & 1 & 1 & 1 & 1 & 1 & 1 & 1 & 2 & 3 & 4 & 5 & 6 & 7 \\
 9 & 1 & 1 & 1 & 1 & 1 & 1 & 1 & 1 & 1 & 1 & 2 & 3 & 4 & 5 & 6 \\
 10 & 1 & 1 & 1 & 1 & 1 & 1 & 1 & 1 & 1 & 1 & 1 & 2 & 3 & 4 & 5 \\
\end{array}
}\)
\captionof{table}{Values of the $h$-Fibonacci sequence
$\mathcal{F}^{(h)}=\{F_n^{(h)}\}_{n\geq 1}$}
\label{tab:F_n}
\end{center}

\begin{center}
\noindent\(\footnotesize{\arraycolsep=2.2pt
\begin{array}{c|cccccccccccccc}
  & \text{n=0} & 1 & 2 & 3 & 4 & 5 & 6 & 7 & 8 & 9 & 10 & 11 & 12 & 13 \\
\hline
 \text{h=0} & 0 & 1 & 4 & 12 & 32 & 80 & 192 & 448 & 1024 & 2304 & 5120 & 11264 & 24576 & 53248 \\
 1 & 0 & 1 & 2 & 5 & 10 & 20 & 38 & 71 & 130 & 235 & 420 & 744 & 1308 & 2285 \\
 2 & 0 & 1 & 2 & 3 & 6 & 11 & 18 & 30 & 50 & 81 & 130 & 208 & 330 & 520 \\
 3 & 0 & 1 & 2 & 3 & 4 & 7 & 12 & 19 & 28 & 42 & 64 & 97 & 144 & 212 \\
 4 & 0 & 1 & 2 & 3 & 4 & 5 & 8 & 13 & 20 & 29 & 40 & 56 & 80 & 115 \\
 5 & 0 & 1 & 2 & 3 & 4 & 5 & 6 & 9 & 14 & 21 & 30 & 41 & 54 & 72 \\
 6 & 0 & 1 & 2 & 3 & 4 & 5 & 6 & 7 & 10 & 15 & 22 & 31 & 42 & 55 \\
 7 & 0 & 1 & 2 & 3 & 4 & 5 & 6 & 7 & 8 & 11 & 16 & 23 & 32 & 43 \\
 8 & 0 & 1 & 2 & 3 & 4 & 5 & 6 & 7 & 8 & 9 & 12 & 17 & 24 & 33 \\
 9 & 0 & 1 & 2 & 3 & 4 & 5 & 6 & 7 & 8 & 9 & 10 & 13 & 18 & 25 \\
 10 & 0 & 1 & 2 & 3 & 4 & 5 & 6 & 7 & 8 & 9 & 10 & 11 & 14 & 19 \\
\end{array}}\)
\captionof{table}{The number $H_n^{(h)}$ of edges of $\hp{n}{h}$}
\label{tab:H_n}
\end{center}

\begin{center}
\noindent\(\footnotesize{\arraycolsep=2.2pt
\begin{array}{c|ccccccccccccccccc}
  & \text{n=0} & 1 & 2 & 3 & 4 & 5 & 6 & 7 & 8 & 9 & 10 & 11 & 12 & 13 & 14 & 15 & 16 \\
\hline
 \text{k=0} & 1 & 1 & 1 & 1 & 1 & 1 & 1 & 1 & 1 & 1 & 1 & 1 & 1 & 1 & 1 & 1 & 1 \\
 1 & 0 & 1 & 2 & 3 & 4 & 5 & 6 & 7 & 8 & 9 & 10 & 11 & 12 & 13 & 14 & 15 & 16 \\
 2 & 0 & 0 & 0 & 0 & 2 & 5 & 9 & 14 & 20 & 27 & 35 & 44 & 54 & 65 & 77 & 90 & 104 \\
 3 & 0 & 0 & 0 & 0 & 0 & 0 & 2 & 7 & 16 & 30 & 50 & 77 & 112 & 156 & 210 & 275 & 352 \\
 4 & 0 & 0 & 0 & 0 & 0 & 0 & 0 & 0 & 2 & 9 & 25 & 55 & 105 & 182 & 294 & 450 & 660 \\
 5 & 0 & 0 & 0 & 0 & 0 & 0 & 0 & 0 & 0 & 0 & 2 & 11 & 36 & 91 & 196 & 378 & 672 \\
 6 & 0 & 0 & 0 & 0 & 0 & 0 & 0 & 0 & 0 & 0 & 0 & 0 & 2 & 13 & 49 & 140 & 336 \\
 7 & 0 & 0 & 0 & 0 & 0 & 0 & 0 & 0 & 0 & 0 & 0 & 0 & 0 & 0 & 2 & 15 & 64 \\
 8 & 0 & 0 & 0 & 0 & 0 & 0 & 0 & 0 & 0 & 0 & 0 & 0 & 0 & 0 & 0 & 0 & 2 \\
\end{array}}\)
\captionof{table}{The number $c_{n,k}^{(1)}$ of  independent $k$-subsets of $\cy{n}{1}$}
\label{tab:q_nk_1}
\end{center}

\begin{center}
\noindent\(\footnotesize{\arraycolsep=2.2pt
\begin{array}{c|cccccccccccccccccc}
  & \text{n=0} & 1 & 2 & 3 & 4 & 5 & 6 & 7 & 8 & 9 & 10 & 11 & 12 & 13 & 14 & 15 & 16 & 17 \\
\hline
 \text{k=0} & 1 & 1 & 1 & 1 & 1 & 1 & 1 & 1 & 1 & 1 & 1 & 1 & 1 & 1 & 1 & 1 & 1 & 1 \\
 1 & 0 & 1 & 2 & 3 & 4 & 5 & 6 & 7 & 8 & 9 & 10 & 11 & 12 & 13 & 14 & 15 & 16 & 17 \\
 2 & 0 & 0 & 0 & 0 & 0 & 0 & 3 & 7 & 12 & 18 & 25 & 33 & 42 & 52 & 63 & 75 & 88 & 102 \\
 3 & 0 & 0 & 0 & 0 & 0 & 0 & 0 & 0 & 0 & 3 & 10 & 22 & 40 & 65 & 98 & 140 & 192 & 255 \\
 4 & 0 & 0 & 0 & 0 & 0 & 0 & 0 & 0 & 0 & 0 & 0 & 0 & 3 & 13 & 35 & 75 & 140 & 238 \\
 5 & 0 & 0 & 0 & 0 & 0 & 0 & 0 & 0 & 0 & 0 & 0 & 0 & 0 & 0 & 0 & 3 & 16 & 51 \\
\end{array}}\)
\captionof{table}{The number $c_{n,k}^{(2)}$ of  independent $k$-subsets of $\cy{n}{2}$}
\label{tab:q_nk_2}
\end{center}

\begin{center}
\noindent\(\footnotesize{\arraycolsep=2.2pt
\begin{array}{c|ccccccccccccccccccc}
  & \text{n=0} & 1 & 2 & 3 & 4 & 5 & 6 & 7 & 8 & 9 & 10 & 11 & 12 & 13 & 14 & 15 & 16 & 17 & 18 \\
\hline
 \text{k=0} & 1 & 1 & 1 & 1 & 1 & 1 & 1 & 1 & 1 & 1 & 1 & 1 & 1 & 1 & 1 & 1 & 1 & 1 & 1 \\
 1 & 0 & 1 & 2 & 3 & 4 & 5 & 6 & 7 & 8 & 9 & 10 & 11 & 12 & 13 & 14 & 15 & 16 & 17 & 18 \\
 2 & 0 & 0 & 0 & 0 & 0 & 0 & 0 & 0 & 4 & 9 & 15 & 22 & 30 & 39 & 49 & 60 & 72 & 85 & 99 \\
 3 & 0 & 0 & 0 & 0 & 0 & 0 & 0 & 0 & 0 & 0 & 0 & 0 & 4 & 13 & 28 & 50 & 80 & 119 & 168 \\
 4 & 0 & 0 & 0 & 0 & 0 & 0 & 0 & 0 & 0 & 0 & 0 & 0 & 0 & 0 & 0 & 0 & 4 & 17 & 45 \\
\end{array}}\)
\captionof{table}{The number $c_{n,k}^{(3)}$ of  independent $k$-subsets of $\cy{n}{3}$}
\label{tab:q_nk_3}
\end{center}
%

\begin{center}
\noindent\(\footnotesize{\arraycolsep=2.2pt
\begin{array}{c|ccccccccccccccccc}
  & \text{n=0} & 1 & 2 & 3 & 4 & 5 & 6 & 7 & 8 & 9 & 10 & 11 & 12 & 13 & 14 & 15 & 16 \\
\hline
 \text{h=0} & 1 & 2 & 4 & 8 & 16 & 32 & 64 & 128 & 256 & 512 & 1024 & 2048 & 4096 & 8192 & 16384 & 32768 & 65536 \\
 1 & 1 & 2 & 3 & 4 & 7 & 11 & 18 & 29 & 47 & 76 & 123 & 199 & 322 & 521 & 843 & 1364 & 2207 \\
 2 & 1 & 2 & 3 & 4 & 5 & 6 & 10 & 15 & 21 & 31 & 46 & 67 & 98 & 144 & 211 & 309 & 453 \\
 3 & 1 & 2 & 3 & 4 & 5 & 6 & 7 & 8 & 13 & 19 & 26 & 34 & 47 & 66 & 92 & 126 & 173 \\
 4 & 1 & 2 & 3 & 4 & 5 & 6 & 7 & 8 & 9 & 10 & 16 & 23 & 31 & 40 & 50 & 66 & 89 \\
 5 & 1 & 2 & 3 & 4 & 5 & 6 & 7 & 8 & 9 & 10 & 11 & 12 & 19 & 27 & 36 & 46 & 57 \\
 6 & 1 & 2 & 3 & 4 & 5 & 6 & 7 & 8 & 9 & 10 & 11 & 12 & 13 & 14 & 22 & 31 & 41 \\
 7 & 1 & 2 & 3 & 4 & 5 & 6 & 7 & 8 & 9 & 10 & 11 & 12 & 13 & 14 & 15 & 16 & 25 \\
 8 & 1 & 2 & 3 & 4 & 5 & 6 & 7 & 8 & 9 & 10 & 11 & 12 & 13 & 14 & 15 & 16 & 17 \\
 9 & 1 & 2 & 3 & 4 & 5 & 6 & 7 & 8 & 9 & 10 & 11 & 12 & 13 & 14 & 15 & 16 & 17 \\
 10 & 1 & 2 & 3 & 4 & 5 & 6 & 7 & 8 & 9 & 10 & 11 & 12 & 13 & 14 & 15 & 16 & 17 \\
\end{array}
}\)
\captionof{table}{The number $c_{n}^{(h)}$ of  all independent sets of $\cy{n}{h}$}
\label{tab:q_nk}
\end{center}

\begin{center}
\noindent\(\footnotesize{\arraycolsep=2.2pt
\begin{array}{c|ccccccccccccccc}
  & \text{n=1} & 2 & 3 & 4 & 5 & 6 & 7 & 8 & 9 & 10 & 11 & 12 & 13 & 14 & 15 \\
\hline
 \text{h=0} & 1 & 2 & 4 & 8 & 16 & 32 & 64 & 128 & 256 & 512 & 1024 & 2048 & 4096 & 8192 & 16384 \\
 1 & 2 & 1 & 3 & 4 & 7 & 11 & 18 & 29 & 47 & 76 & 123 & 199 & 322 & 521 & 843 \\
 2 & 3 & 1 & 1 & 4 & 5 & 6 & 10 & 15 & 21 & 31 & 46 & 67 & 98 & 144 & 211 \\
 3 & 4 & 1 & 1 & 1 & 5 & 6 & 7 & 8 & 13 & 19 & 26 & 34 & 47 & 66 & 92 \\
 4 & 5 & 1 & 1 & 1 & 1 & 6 & 7 & 8 & 9 & 10 & 16 & 23 & 31 & 40 & 50 \\
 5 & 6 & 1 & 1 & 1 & 1 & 1 & 7 & 8 & 9 & 10 & 11 & 12 & 19 & 27 & 36 \\
 6 & 7 & 1 & 1 & 1 & 1 & 1 & 1 & 8 & 9 & 10 & 11 & 12 & 13 & 14 & 22 \\
 7 & 8 & 1 & 1 & 1 & 1 & 1 & 1 & 1 & 9 & 10 & 11 & 12 & 13 & 14 & 15 \\
 8 & 9 & 1 & 1 & 1 & 1 & 1 & 1 & 1 & 1 & 10 & 11 & 12 & 13 & 14 & 15 \\
 9 & 10 & 1 & 1 & 1 & 1 & 1 & 1 & 1 & 1 & 1 & 11 & 12 & 13 & 14 & 15 \\
 10 & 11 & 1 & 1 & 1 & 1 & 1 & 1 & 1 & 1 & 1 & 1 & 12 & 13 & 14 & 15 \\
\end{array}
}\)
\captionof{table}{Values of the $h$-Lucas sequence
$\mathcal{L}^{(h)}=\{L_n^{(h)}\}_{n\geq 1}$}
\label{tab:L_n}
\end{center}

\begin{center}
\noindent\(\footnotesize{\arraycolsep=2.2pt
\begin{array}{c|cccccccccccccccc}
  & \text{n=0} & 1 & 2 & 3 & 4 & 5 & 6 & 7 & 8 & 9 & 10 & 11 & 12 & 13 & 14 & 15 \\
\hline
 \text{h=0} & 0 & 1 & 4 & 12 & 32 & 80 & 192 & 448 & 1024 & 2304 & 5120 & 11264 & 24576 & 53248 & 114688 & 245760 \\
 1 & 0 & 0 & 2 & 3 & 8 & 15 & 30 & 56 & 104 & 189 & 340 & 605 & 1068 & 1872 & 3262 & 5655 \\
 2 & 0 & 0 & 0 & 3 & 4 & 5 & 12 & 21 & 32 & 54 & 90 & 143 & 228 & 364 & 574 & 900 \\
 3 & 0 & 0 & 0 & 0 & 4 & 5 & 6 & 7 & 16 & 27 & 40 & 55 & 84 & 130 & 196 & 285 \\
 4 & 0 & 0 & 0 & 0 & 0 & 5 & 6 & 7 & 8 & 9 & 20 & 33 & 48 & 65 & 84 & 120 \\
 5 & 0 & 0 & 0 & 0 & 0 & 0 & 6 & 7 & 8 & 9 & 10 & 11 & 24 & 39 & 56 & 75 \\
 6 & 0 & 0 & 0 & 0 & 0 & 0 & 0 & 7 & 8 & 9 & 10 & 11 & 12 & 13 & 28 & 45 \\
 7 & 0 & 0 & 0 & 0 & 0 & 0 & 0 & 0 & 8 & 9 & 10 & 11 & 12 & 13 & 14 & 15 \\
 8 & 0 & 0 & 0 & 0 & 0 & 0 & 0 & 0 & 0 & 9 & 10 & 11 & 12 & 13 & 14 & 15 \\
 9 & 0 & 0 & 0 & 0 & 0 & 0 & 0 & 0 & 0 & 0 & 10 & 11 & 12 & 13 & 14 & 15 \\
 10 & 0 & 0 & 0 & 0 & 0 & 0 & 0 & 0 & 0 & 0 & 0 & 11 & 12 & 13 & 14 & 15 \\
\end{array}
}\)
\captionof{table}{The number $M_n^{(h)}$ of edges of $\hc{n}{h}$, for $n>h$}
\label{tab:M_n}
\end{center}

\section*{Acknowledgment}
We are grateful to the editor for his precious work, and to the anonymous referee for a
careful reading of our paper,
and for his/her helpful suggestions which allowed us to improve the presentation,
and notation, of this work.

\bibliographystyle{elsarticle-num}

\end{document}